\documentclass[11pt,a4paper,titlepage]{article}
\usepackage[utf8]{inputenc}
\usepackage[T1]{fontenc}
\usepackage[english]{babel}
\usepackage{mathrsfs}
\usepackage{amsmath}
\usepackage{amsfonts}
\usepackage{amssymb}
\usepackage{amsthm}
\usepackage{longtable}
\usepackage{enumerate}
\usepackage{bussproofs}
\usepackage{proof}
\usepackage{xcolor}
\usepackage[width=0.00cm, height=2.50cm, left=2.50cm, right=3.00cm, top=2.50cm, bottom=2.50cm]{geometry}
\usepackage{url}

\theoremstyle{definition}
\newtheorem{defn}{Definition}[section]
\newtheorem{theorem}{\textbf{Theorem}}

\newtheorem{proposition}{ \textbf{Proposition}}

\usepackage{units}
\usepackage{bussproofs}
\newcommand{\var}{\mathrm{var}}
\usepackage[normalem]{ulem}

\newcommand{\half}{\nicefrac[]{1}{2}}

\usepackage{longtable,enumerate}

\usepackage[breaklinks,colorlinks,citecolor=blue,urlcolor=blue]{hyperref}
\newcommand\DOI[2]{DOI: \href{#1#2}{#2}}

\newenvironment{Published}{\begin{list}{}{\leftmargin=25pt}
    \item[]\small\noindent{Published as}}%
    {\end{list}}
    \newenvironment{Funding}{\begin{list}{}{\leftmargin=25pt}
        \item[]\small\noindent{The research}}%
        {\end{list}}
\newenvironment{Abstract}{\begin{list}{}{\leftmargin=25pt}
    \item[]\footnotesize\noindent{\sc \bf Abstract:}}%
{\end{list}}
\newenvironment{keywords}{\begin{list}{}{\leftmargin=25pt}
    \item[]\footnotesize\noindent{\sc \bf Keywords:}}%
{\end{list}}

\begin{document}
\begin{center}
\vspace{0.3cm}
{\large  
Yaroslav Petrukhin\vspace{0.3cm}\\
\textbf{On Three-Valued Dependence-Like Logics}}

\ \\
\begin{footnotesize}
 Center for Philosophy of Nature, University of Lodz, Poland,  \url{yaroslav.petrukhin@gmail.com}\\
\end{footnotesize}
\end{center}
\begin{Published}
Petrukhin, Y. (2026)
On Three-Valued Dependence-Like Logics Logic and Logical Philosophy. Online first paper. 
\DOI{http://dx.doi.org/}{10.12775/LLP.2026.016}
\bigskip
\end{Published}
\begin{Funding}
 was funded by the European Union (ERC, ExtenDD, project number: 101054714). Views and opinions expressed are however those of the author(s) only and do not necessarily reflect those of the European Union or the European Research Council. Neither the European Union nor the granting authority can be held responsible for them. 
\end{Funding}
 \begin{Abstract}
In this paper, we investigate connections between dependence logics, viewed as a subclass of relating logics, and three-valued logics. More specifically, we identify common features of Epstein's subject-matter semantics and the variable-inclusion conditions characteristic of some infectious many-valued logics. Inspired by Del Cerro and Lugardon's sequent calculi for dependence logics, in which classical connectives are combined with connectives satisfying Epstein-style conditions, we introduce two three-valued dependence-like logics that combine classical conjunction and disjunction with infectious negation and implication. We also provide sound, complete, and cut-free bisequent calculi for these logics.
 \end{Abstract}

\section{Introduction}
Dependence logics were introduced by Epstein \cite{Epstein} and may
be regarded as a special class of relating logics, although the broader programme of relating logics was articulated later (see, e.g., Jarmu\.{z}ek's paper \cite{Jarmuzek}). One of these dependence logics, known as \textbf{DAI} (the logic of demodalised analytic implication), was originally introduced by Dunn \cite{Dunn} in the context of relevance logics. Some connections between three- and four-valued logics and dependence
 logics were established by Demolombe \cite{Demolombe}. Later on, \textbf{DAI} was studied by Pra Baldi \cite{Baldi} in the context of logics of variable inclusion, in particular those known as infectious or nonsense logics (a special class of many-valued ones). Moreover, Pra Baldi presented a matrix semantics based on P\l{}onka sums for \textbf{DAI} (see also the paper \cite{LeddaPaoliBaldi} by Ledda, Paoli, and Pra Baldi). In this paper, we continue this line of research and introduce two simple dependence-like three-valued logics inspired by Del Cerro and Lugardon's sequent-calculus approach \cite{SequentsDependence} to dependence logics and their relatives.

 One of Epstein's ideas is that $ v(A\rightarrow B)=1 $ under a valuation $ v $ iff not only $ v(A)=0 $ or $ v(B)=1 $, but also $ s(B)\subseteq s(A) $, where $ s $ is a function from the set of all propositional variables to a set of `subject-matters'. This structure is understood rather broadly, which, we believe, leaves room for different interpretations. Following the approaches of Parry \cite{Parry}, Dunn \cite{Dunn}, and Epstein \cite{Epstein87,Epstein}, Pra Baldi \cite{Baldi} adopted the following reading of the expression `$ s(B)\subseteq s(A) $': every propositional variable occurring in $B$ also occurs in $A$. Thus, dependence logics may be connected with  logics of variable inclusion (some of which are many-valued). Pra Baldi's interpretation was inspired by Parry's
\cite{Parry} variable-inclusion requirement, which may be formulated
as the so-called \emph{Proscriptive Principle} (PP, for short; see
\cite{Ferguson} for further discussion). According to PP, every
propositional variable occurring in the conclusion must also occur
in the premises. This principle is widely discussed in the context of infectious logics, i.e. logics with an infectious value such that, when this value is
assigned to an argument of an operation, the operation also takes this value. Additionally, this principle might be viewed as a stronger version of the variable sharing principle, one of the main properties of relevance logics.\footnote{``The variable sharing principle says that no formula of the form $ A\rightarrow B$ can be proven in a relevance logic if $A$ and $B$ do not have at least one propositional variable (sometimes called a proposition letter) in common and that no inference can be shown valid if the premises and conclusion do not share at least one propositional variable.'' \cite{Stanford}.}

As shown by Ciuni and Carrara \cite{CiuniCarrara}, if PP holds and
a set of premises $\Gamma$ classically entails a formula $A$, then
$\Gamma$ entails $A$ in weak Kleene logic $\bf K_3^w$ \cite{Kleene38}. This principle might be viewed as one possible specification of Epstein's idea that $ v(A\rightarrow B)=1 $ under a valuation $ v $ iff not only $ v(A)=0 $ or $ v(B)=1 $, but also $ s(B)\subseteq s(A) $, where $ s $ is a function from the set of all propositional variables to a set of `subject-matters' (see Definition \ref{s} for its formal presentation).

On the other hand, one may also consider the converse of PP, which requires that all propositional variables of the premises must be included in the set of all propositional variables of the conclusion. As was shown by Ciuni and Carrara \cite{CiuniCarrara}, if such a principle holds and $ \Gamma $  classically entails $ A $, then $ \Gamma $ entails $ A $ in paraconsistent weak Kleene logic $ \bf PWK$ \cite{Hallden}. Epstein also considered a dependence logic whose implication requires
$s(A)\subseteq s(B)$ rather than $s(B)\subseteq s(A)$.

Del Cerro and Lugardon \cite{SequentsDependence} formulated sequent calculi for some dependence logics. Their rules for conjunction and disjunction are classical, while the ones for the implication and negation are modifications of the classical rules. Taking these observations into account, including the connection between
Epstein's subject-matter conditions and PP, the relationship between
dependence and many-valued logics, and Del Cerro and Lugardon's
proof-theoretic approach, we propose two three-valued logics with
classical conjunction and disjunction and non-classical implication
and negation. The first logic takes its infectious implication and
negation from $\bf K_3^w$, which is connected with PP, whereas
the second takes them from \textbf{PWK}, which is connected with the
converse variable-inclusion condition. Our logics are combinations of the connectives of strong Kleene logic $ \bf K_3$ \cite{Kleene38} and the logic of paradox \textbf{LP} \cite{Asenjo66,LP} with the connectives of weak Kleene logic $ \bf K_3^w$ \cite{Kleene38} and its paraconsistent version \textbf{PWK} \cite{Hallden}.  

The structure of the rest of our paper is as follows. In Section \ref{semantics}, we analyse the connection between many-valued and dependence logics and introduce our logics. In Section \ref{proof}, we present cut-free bisequent calculi for them. Section \ref{conclusion} contains the conclusion.

\section{From semantics of dependence logics to semantics of our logics}\label{semantics}
Let us fix a propositional language $ \mathscr{L} $ with the alphabet $ \langle \mathcal{P},\wedge,\vee,\rightarrow,\neg,(,)\rangle $, where $\mathcal{P}=\{p_1,q_1,r_1,s_1,p_2,\ldots\}$ is the set of propositional variables. The set $\mathscr{F}$ of all $\mathscr{L}$-formulas is defined inductively in the
usual way. We write $A,B,C,A_1,\ldots$ for the metavariables ranging over formulas and $\Gamma,\Delta,\Pi,\Sigma,\Gamma',\ldots$ for the metavariables ranging over sets of formulas. We write $\var(A)$ for the set of all propositional variables of a formula $A$. Let $ \var(\Gamma) $ be the set of all propositional variables of a set of formulas $ \Gamma $. Before describing the semantics of dependence logics, let us emphasize
the main idea behind these logics, as Del Cerro and Lugardon \cite{SequentsDependence} put it:
\begin{quote}
	``The idea behind dependence logics is to consider that the conditional
	expression `$A\rightarrow B$' is true if and only if $A$ is false or $B$ is true, and the
	subject matters of $A$ and $B$ are related.'' \cite[p. 58, notation adjusted]{SequentsDependence}
\end{quote}

The semantics of dependence logics can be described as follows (we follow Del Cerro and Lugardon's \cite[p. 58]{SequentsDependence} presentation).

\begin{defn}\label{s}
Let $ s $ be a function from $ \mathcal{P} $ into the power set of $\mathcal{T}$ (understood as the set of topics or `subject-matters'). Then $s$ is extended to
formulas and sets of formulas by the following clauses:
\begin{itemize}\itemsep=0pt
	\item $s(A)=\bigcup_{p\in\var(A)}s(p)$, 
	\item $s(\Gamma)=\bigcup_{A\in\Gamma}s(A)$.
\end{itemize}	
A valuation, relative to $s$, is a function $v_s$ from $\mathcal{P}$ into the set of truth
values $\{1, 0\}$. The valuation $v_s$ of formulas is defined for negation, conjunction, and disjunction,
as in classical logic. For implication, $v_s$ is defined as follows (several options are possible):
\begin{enumerate}\itemsep=0pt
	\item $ v_s(A\rightarrow B)=1 $ iff \big($v_s(A)=0$ or $v_s(B)=1$\big) and $ s(B)\subseteq s(A)$,
	\item $ v_s(A\rightarrow B)=1 $ iff \big($v_s(A)=0$ or $v_s(B)=1$\big) and $ s(A)\subseteq s(B)$,
	\item $ v_s(A\rightarrow B)=1 $ iff \big($v_s(A)=0$ or $v_s(B)=1$\big) and $ s(A)=s(B)$.
\end{enumerate}
For a given function $s$, a formula $A$ is $s$-valid iff
$v_s(A)=1$ for every valuation $v_s$.
\end{defn}

The first option is Epstein's original interpretation \cite{Epstein},
and the second is its converse. In what follows, we are going to focus on these two interpretations, since they may be connected with three-valued logics. With the proscriptive principle and its converse in mind, we could think about the following interpretations:

\begin{enumerate}\itemsep=0pt
\item $ v(A\rightarrow B)=1 $ iff \big($v(A)=0$ or $v(B)=1$\big) and $ \var(B)\subseteq\var(A)$,
\item $ v(A\rightarrow B)=1 $ iff \big($v(A)=0$ or $v(B)=1$\big) and $ \var(A)\subseteq\var(B)$.
\end{enumerate}

We could also think about some three-valued interpretations of implication which do not postulate the variable inclusion conditions, but have them as a consequence. One may find such a definition among infectious logics. However, to make this idea precise, we first need some formal definitions.

\begin{defn}
By a logical matrix we mean a triple $ \langle V,C,D\rangle $, where $ V $ is the set of truth values, $ C $ is the set of functions corresponding to connectives, and $ D $ is the set of designated values. The entailment relation is defined as follows:
\begin{center}
$\Gamma\models\Delta$ iff, for every valuation $v$, if
$v(A)\in D$ for every $A\in\Gamma$, then $v(B)\in D$ for some
$B\in\Delta$. 
\end{center}
\end{defn}

\begin{defn}
By a logic \textbf{L} we mean the set $ \{\langle \Gamma,\Delta\rangle: \Gamma\models_{\bf L}\Delta \} $.
\end{defn}

\begin{defn}
An element $i\in V$ is an infectious value of a logical matrix $\langle V,C,D\rangle$ iff for every $n$-ary function $c\in C$ and every tuple
$\langle x_1,\ldots,x_n\rangle\in V^n$, if $x_j=i$ for some $1\leq j\leq n$,
then $c(x_1,\ldots,x_n)=i$.
\end{defn}
 
\begin{defn}
 A logic having a logical matrix with an infectious value is said to be infectious.
\end{defn}

\begin{theorem}[Ciuni, Ferguson, and Szmuc \cite{CiuniFergusonSzmuc}]
Let $\mathbf{L}$ be the logic determined by the matrix
$\mathcal{M}_{\mathbf{L}}=\langle V,C,D\rangle$, and let $i\notin V$.
Assume that every function in $C$ has positive finite arity. For every
$k$-ary function $c\in C$, let $c^{i}$ be the $k$-ary function on
$V\cup\{i\}$ defined by $c^{i}(x_{1},\ldots,x_{k})=i$ if
$x_{j}=i$ for some $1\leq j\leq k$, and
$c^{i}(x_{1},\ldots,x_{k})=c(x_{1},\ldots,x_{k})$ otherwise.
Let $C^{i}=\{c^{i}:c\in C\}$. Let $\mathbf{S}$ and $\mathbf{T}$ be the logics determined by the
matrices $\mathcal{M}_{\mathbf{S}}=\langle V\cup\{i\},C^{i},D\rangle$
and $\mathcal{M}_{\mathbf{T}}=\langle V\cup\{i\},C^{i},D\cup\{i\}\rangle$, respectively. Then:
\begin{itemize}\itemsep=0pt
    \item $\Gamma\models_{\mathbf{S}}\Delta$ iff there is some
    $\Delta'\subseteq\Delta$ such that
    $\operatorname{var}(\Delta')\subseteq\operatorname{var}(\Gamma)$
    and $\Gamma\models_{\mathbf{L}}\Delta'$;

    \item $\Gamma\models_{\mathbf{T}}\Delta$ iff there is some
    $\Gamma'\subseteq\Gamma$ such that
    $\operatorname{var}(\Gamma')\subseteq\operatorname{var}(\Delta)$
    and $\Gamma'\models_{\mathbf{L}}\Delta$.
\end{itemize}
\end{theorem}
 
Let us present some examples of infectious logics. Here are the matrices of weak Kleene logic $ \bf K_3^w$ \cite{Kleene38} and its paraconsistent version \textbf{PWK} \cite{Hallden}.

\begin{center}
{\small  \begin{tabular}{|c|c|}\hline
$A $ & $\neg $  \\\hline
1 & 0  \\\hline
\nicefrac[]{1}{2} & \nicefrac[]{1}{2}  \\\hline
0 & 1 \\\hline
\end{tabular} 
             \begin{tabular}{|c|ccc|}\hline
              $ \vee $ & 1 & \nicefrac[]{1}{2} &  0 \\\hline
              1 & 1 & \nicefrac[]{1}{2} & 1 \\\hline
            \nicefrac[]{1}{2} & \nicefrac[]{1}{2} & \nicefrac[]{1}{2} & \nicefrac[]{1}{2}  \\\hline
              0 & 1 & \nicefrac[]{1}{2} & 0 \\\hline
              \end{tabular}
                \begin{tabular}{|c|ccc|}\hline
                 $ \wedge $ & 1 & \nicefrac[]{1}{2} & 0 \\\hline
                 1 & 1 & \nicefrac[]{1}{2} & 0 \\\hline
               \nicefrac[]{1}{2} & \nicefrac[]{1}{2} & \nicefrac[]{1}{2} & \nicefrac[]{1}{2} \\\hline
                 0 & 0 & \nicefrac[]{1}{2} & 0 \\\hline
                 \end{tabular}
                \begin{tabular}{|c|ccc|}\hline
                 $ \rightarrow $ & 1 & \nicefrac[]{1}{2} & 0 \\\hline
1 & 1 & \nicefrac[]{1}{2} & 0 \\\hline
\nicefrac[]{1}{2} & \nicefrac[]{1}{2} & \nicefrac[]{1}{2} & \nicefrac[]{1}{2} \\\hline
                 0 & 1 & \nicefrac[]{1}{2} & 1 \\\hline
                 \end{tabular}                  }
\end{center}

These logics have the same connectives, but different sets of designated values: in the case of $\bf K_3^w$, $D=\{1\}$, and in the case of $\bf PWK$, $D=\{1,\nicefrac[]{1}{2}\}$. 
We can characterize the entailment relations of $ \bf K_3^w $ and \textbf{PWK} as a corollary of the theorem of Ciuni, Ferguson, and Szmuc (\textbf{CL} stands for classical logic):\footnote{These results were first obtained by Ciuni and Carrara \cite{CiuniCarrara}.}

\begin{itemize}\itemsep=0pt
\item $ \Gamma\models_{\bf K_3^w}\Delta $ iff $ \var(\Delta^\prime)\subseteq\var(\Gamma) $ for some $ \Delta^\prime\subseteq\Delta $ such that $ \Gamma\models_{\bf CL}\Delta^\prime $,
\item $ \Gamma\models_{\bf PWK}\Delta $ iff $ \var(\Gamma^\prime)\subseteq\var(\Delta) $ for some $ \Gamma^\prime\subseteq\Gamma $ such that $ \Gamma^\prime\models_{\bf CL}\Delta $,
\item If $ \var(\Delta)\subseteq\var(\Gamma) $ and $ \Gamma\models_{\bf CL}\Delta $, then $ \Gamma\models_{\bf K_3^w}\Delta $,
\item If $ \var(\Gamma)\subseteq\var(\Delta) $ and $ \Gamma\models_{\bf CL}\Delta $, then $ \Gamma\models_{\bf PWK}\Delta $.
\end{itemize}

Thus, in infectious logics, the connectives are defined in such a way that their entailment relations can be characterised via the entailment relations of some non-infectious logics, but subject to restrictions on propositional variables. These restrictions might be viewed as a particular case of the conditions imposed on the function $s$ in
dependence logics. Taking this observation into account, we do not think it would be
correct to say that infectious logics are dependence logics.
Nevertheless, there is a clear connection between these types of logics. We may say that the case of $ \bf K_3^w$ in a sense reflects the condition $ s(B)\subseteq s(A) $ from Epstein's definition of implication. The case of \textbf{PWK} reflects the condition $ s(A)\subseteq s(B) $ from the converse of Epstein's definition. To sum up, we state the following claims:
\begin{itemize} \itemsep=0pt
\item Variable inclusion conditions may be viewed as a particular specification of Epstein's conditions used in the definition of an implication.
\item Infectious logics might be viewed as relatives of dependence logics and as a kind of simplification of the idea of dependence logics, but probably not as dependence logics themselves.

\item In dependence logics, restrictions regarding common content are used for implication only, while in infectious logics variable inclusion conditions are used for all the connectives.
\end{itemize}

To pursue the connection between infectious and dependence logics further, we consider two options:

\begin{itemize}\itemsep=0pt
\item To consider logics with classical negation, conjunction, and disjunction, but with infectious implication.
\item To let negation be infectious as well, taking into account Del Cerro and Lugardon's work \cite{SequentsDependence}, where sequent rules for negation are non-classical (notice that their calculi make it possible to prove decidability of Epstein's logics, but are complete for certain modifications of Epstein's logics).
\end{itemize}

Let us briefly look at the first option. We need to determine which three-valued negations, conjunctions, and
disjunctions count as classical. 
According to Rosser and Turquette \cite{RosserTurkette}, these are
the connectives satisfying the following standard conditions:
\begin{itemize}\itemsep=0pt
\item $ v(A\wedge B)\in D$ iff $ v(A)\in D$ and $ v(B)\in D$,
\item $ v(A\vee B)\not\in D$ iff $ v(A)\not\in D$ and $ v(B)\not\in D$,
\item $ v(\neg A)\in D $ iff $ v(A)\not\in D$,
\item $ v(A\rightarrow B)\notin D$ iff $ v(A)\in D$ and $ v(B)\not\in D$.
\end{itemize}

One may easily check that the implication shared by $ \bf K_3^w$ and \textbf{PWK} (weak Kleene \cite{Kleene38} implication, studied also by Bochvar \cite{Bochvar}) is the only three-valued infectious implication that coincides with classical implication on the set $ \{0,1\} $. Thus, each triple $\langle\neg,\land,\lor\rangle$ of connectives satisfying
the standard conditions, together with weak Kleene--Bochvar implication,
determines a three-valued dependence-like logic.  
 
We may also drop the requirement that the connectives coincide with
their classical counterparts on the set $\{0,1\}$, thereby obtaining
an even wider class of connectives.

Let us now consider the second option, in which negation is infectious. The negation shared by $\bf K_3^w$ and \textbf{PWK} is the only three-valued infectious negation that coincides with classical negation
on $\{0,1\}$. 
Such a choice reduces the set of candidates for a three-valued dependence-like logic. We have one candidate for negation, one candidate for implication, and several candidates for conjunction and disjunction. We can vary the set of designated values: either $ \{1\} $ or $ \{1,\half\} $. More generally, we may consider connectives that do not coincide with
their classical counterparts on $\{0,1\}$. We require only that they
take a designated value whenever the corresponding classical connective
takes the value $1$, and a non-designated value whenever it takes the
value $0$. As a result, we get a wide family of dependence-like logics.

Within this family, the two most interesting logics are those whose implication and negation are weak Kleene--Bochvar and whose conjunction and disjunction are given by minimum and maximum, respectively (we get two logics from one set of connectives by changing the set of designated values). These definitions of conjunction and disjunction are the most popular in three-valued logic and are arguably the most natural. It therefore seems reasonable to choose them for our dependence-like logics. 
The truth tables for our logics are as follows: 

\begin{center}
 \begin{tabular}{|c|c|}\hline
$A $ & $\neg $  \\\hline
1 & 0  \\\hline
\nicefrac[]{1}{2} & \nicefrac[]{1}{2}  \\\hline
0 & 1 \\\hline
\end{tabular} 
             \begin{tabular}{|c|ccc|}\hline
              $ \vee $ & 1 & \nicefrac[]{1}{2} &  0 \\\hline
              1 & 1 & 1 & 1 \\\hline
            \nicefrac[]{1}{2} & 1 & \nicefrac[]{1}{2} & \nicefrac[]{1}{2}  \\\hline
              0 & 1 & \nicefrac[]{1}{2} & 0 \\\hline
              \end{tabular}
                \begin{tabular}{|c|ccc|}\hline
                 $ \wedge $ & 1 & \nicefrac[]{1}{2} & 0 \\\hline
                 1 & 1 & \nicefrac[]{1}{2} & 0 \\\hline
               \nicefrac[]{1}{2} & \nicefrac[]{1}{2} & \nicefrac[]{1}{2} & 0 \\\hline
                 0 & 0 & 0 & 0 \\\hline
                 \end{tabular}
                \begin{tabular}{|c|ccc|}\hline
                 $ \rightarrow $ & 1 & \nicefrac[]{1}{2} & 0 \\\hline
1 & 1 & \nicefrac[]{1}{2} & 0 \\\hline
\nicefrac[]{1}{2} & \nicefrac[]{1}{2} & \nicefrac[]{1}{2} & \nicefrac[]{1}{2} \\\hline
                 0 & 1 & \nicefrac[]{1}{2} & 1 \\\hline
                 \end{tabular}                  
\end{center}

Our logics are obtained by replacing the strong Kleene implication in strong Kleene $ \bf K_3$ \cite{Kleene38} and the logic of paradox \textbf{LP} \cite{Asenjo66,LP,Priest} with weak Kleene--Bochvar implication. We call our logics $ \bf K_3^d $ and $ \bf LP^d $ (\textbf{d} stands for dependence) and define their entailment relations as follows:
\begin{itemize}\itemsep=0pt
\item $\Gamma\models_{\bf K_3^d}\Delta$ iff, for every valuation
$v$, if $v(A)=1$ for every $A\in\Gamma$, then $v(B)=1$ for some
$B\in\Delta$;

\item $\Gamma\models_{\bf LP^d}\Delta$ iff, for every valuation
$v$, if $v(A)\neq 0$ for every $A\in\Gamma$, then $v(B)\neq 0$ for
some $B\in\Delta$.
\end{itemize}

At the end of this section, let us investigate some properties of $ \bf K_3^d$ and $ \bf LP^d$. 

\begin{proposition}\label{Express}
The following definability relations hold in
$ \bf K_3^d$ and $ \bf LP^d$, for every valuation $v$
and all formulas $A$ and $B$:
\begin{itemize}\itemsep=0pt
	\item $ v(A\wedge B)=v(\neg(\neg A\vee\neg B)) $,
	\item $ v(A\vee B)=v(\neg(\neg A\wedge\neg B)) $,
	\item $ v(A\rightarrow B)=v((\neg A\vee B)\wedge(A\vee\neg A)\wedge (B\vee\neg B))$.
\end{itemize}	
\end{proposition}
\begin{proof}
	By a routine check.
\end{proof}

Following Priest \cite{Priest}, we say that a logic \textbf{L} is paraconsistent if for some formulas $ A $ and $ B $, it holds that $ A,\neg A\not\models_{\bf L} B$. Following Sette and Carnielli \cite{SetteCarnielli}, we say that a logic \textbf{L} is paracomplete if for some formula $ A $, it holds that $\not\models_{\bf L}A\vee\neg A$.

\begin{proposition}
	\begin{enumerate}[(1)]\itemsep=0pt
		\item $ \bf K_3^d$ is paracomplete,
		\item $ \bf K_3^d$ does not have any tautologies,
		\item $ \bf LP^d$ is paraconsistent,
		\item $\bf LP^d$ and \textbf{CL} have the same tautologies.
	\end{enumerate}	
\end{proposition}
\begin{proof}
	(1) Follows from the fact that $ \not\models p_1\vee\neg p_1$ under the valuation $ v $ such that $ v(p_1)=\half $.
	
	(2)	Let $v(p)=\half$ for every propositional variable $p$.
A straightforward induction on formulas shows that $v(A)=\half$
for every formula $A$. Since $\half$ is not designated in
$\bf K_3^d$, the logic has no tautologies.
	
	(3) Follows from the fact that $p_1,\neg p_1\not\models p_2$ under the valuation $ v $ such that $ v(p_1)=\half $ and $ v(p_2)=0 $.	

(4) We adopt Epstein's proof \cite{Epstein} of the fact that
\textbf{LP} and \textbf{CL} have the same set of tautologies.\footnote{This
result was first obtained by Martin \cite{Martin}.} Let us prove that a formula $F$ is not valid in \textbf{CL} iff there is a valuation $v$ in $\bf LP^d$ such that $v(F)=0$. Observe that $(\ast_1)$ any valuation in \textbf{CL} can also be regarded
as a valuation in $\bf LP^d$. Thus, $(\ast_2)$ if, under some
classical valuation $v'$, we have $v'(F)=0$, then we also have
$v'(F)=0$ in $\bf LP^d$.

Conversely, suppose that there is a valuation $v$ in $\bf LP^d$
such that $v(F)=0$. Let us define a valuation $w$ in \textbf{CL} as
follows. For any $p\in\mathcal{P}$, let $w(p)=1$ iff
$v(p)\in\{1,\half\}$, and let $w(p)=0$ iff $v(p)=0$. By structural
induction on $A$, it is possible to show that:
\begin{enumerate}\itemsep=0pt
    \item[(a)] if $v(A)=0$, then $w(A)=0$;
    \item[(b)] if $v(A)=1$, then $w(A)=1$.
\end{enumerate}

We prove (a) and (b) simultaneously by induction on the complexity of
$A$. As an example, we consider the case of implication. Suppose that
$A=B\rightarrow C$. By the induction hypothesis, $v(B)=0$ implies
$w(B)=0$, $v(C)=0$ implies $w(C)=0$, $v(B)=1$ implies $w(B)=1$, and
$v(C)=1$ implies $w(C)=1$.

Assume that $v(B\rightarrow C)=0$. Then $v(B)=1$ and $v(C)=0$. Hence,
by the induction hypothesis, $w(B)=1$ and $w(C)=0$. Therefore,
$w(B\rightarrow C)=0$. Thus, $v(B\rightarrow C)=0$ implies
$w(B\rightarrow C)=0$.

Assume that $v(B\rightarrow C)=1$. Then either
$v(B)=1$ and $v(C)=1$, or $v(B)=0$ and $v(C)=1$, or $v(B)=0$ and
$v(C)=0$. Suppose first that $v(B)=1$ and $v(C)=1$. By the induction
hypothesis, $w(B)=1$ and $w(C)=1$. Hence, $w(B\rightarrow C)=1$.
Suppose next that $v(B)=0$ and $v(C)=1$. By the induction hypothesis,
$w(B)=0$ and $w(C)=1$. Thus, $w(B\rightarrow C)=1$. Finally, suppose
that $v(B)=0$ and $v(C)=0$. By the induction hypothesis, $w(B)=0$ and
$w(C)=0$. Consequently, $w(B\rightarrow C)=1$. Therefore,
$v(B\rightarrow C)=1$ implies $w(B\rightarrow C)=1$.

The remaining cases are similar. Since $v(F)=0$, condition (a) yields
$w(F)=0$. Hence, $\not\models_{\mathbf{CL}}F$.
\end{proof}

Let us compare $ \bf K_3$ with $ \bf K_3^d$ and \textbf{LP} with $ \bf LP^d$. We know that both $ \bf K_3$ and $ \bf K_3^d$ have no tautologies, but are their consequence relations different? Are the consequence relations of $ \bf LP$ and $ \bf LP^d$ different? Both questions have affirmative answers. Thus, $\bf K_3$ and
$\bf K_3^d$ are distinct, as are \textbf{LP} and $\bf LP^d$.
\begin{proposition}
\begin{itemize}\itemsep=0pt
	\item $\neg p_1\vee p_2\models_{\bf K_3}p_1\rightarrow p_2$, 
	\item $\neg p_1\vee p_2\not\models_{\bf K_3^d}p_1\rightarrow p_2$,
	\item $\neg (p_1\rightarrow p_2)\models_{\bf LP}\neg(\neg p_1\vee p_2)$, 
	\item $\neg (p_1\rightarrow p_2)\not\models_{\bf LP^d}\neg(\neg p_1\vee p_2)$.
\end{itemize}	
\end{proposition}
\begin{proof}
	By a routine check.
\end{proof}
\section{Proof theory for $\bf K_3^d$ and $\bf LP^d$}\label{proof}
In this section, we present cut-free bisequent calculi for $\bf K_3^d$ and $\bf LP^d$ obtained using the uniform approach developed in \cite{IP1,IP2}.

A sequent is an ordered pair written as $\Gamma\Rightarrow\Delta $, where $\Gamma$ and $\Delta$ are finite, possibly empty, multisets of formulas. A bisequent is an ordered pair of sequents $\Gamma\Rightarrow\Delta \mid \Pi\Rightarrow\Sigma$, where $\Gamma, \Delta, \Pi$, and $\Sigma$ are finite, possibly empty, multisets of formulas. In what follows, $ S $ stands for an arbitrary sequent.

Let us describe a bisequent calculus \textbf{BSC}. A bisequent $\Gamma\Rightarrow\Delta \mid \Pi\Rightarrow\Sigma$ is axiomatic iff at least one of
$\Gamma\cap\Sigma$, $\Gamma\cap\Delta$, and $\Pi\cap\Sigma$
is nonempty. The logical rules are as follows:
\begin{center}
\noindent\begin{tabular}{ll}		
	$(\mathord{\neg}\mathord{\Rightarrow\mid})$ \ $\dfrac{\Gamma
		\Rightarrow \Delta \mid \Pi\Rightarrow\Sigma, A}{\neg A, \Gamma \Rightarrow \Delta \mid \Pi\Rightarrow\Sigma}$ &
	$(\mathord{\Rightarrow}\mathord{\neg\mid})$ \ $\dfrac{\Gamma
		\Rightarrow \Delta \mid A, \Pi\Rightarrow\Sigma}{\Gamma
		\Rightarrow \Delta, \neg A \mid \Pi\Rightarrow\Sigma}$\\[16pt]
	
	$(\mathord{\mid\neg}\mathord{\Rightarrow})$ \ $\dfrac{\Gamma
		\Rightarrow \Delta, A \mid \Pi\Rightarrow\Sigma}{\Gamma \Rightarrow \Delta \mid \neg A, \Pi\Rightarrow\Sigma}$ &
	$(\mathord{\mid\Rightarrow}\mathord{\neg})$ \ $\dfrac{A, \Gamma
		\Rightarrow \Delta \mid \Pi\Rightarrow\Sigma}{\Gamma
		\Rightarrow \Delta \mid \Pi\Rightarrow\Sigma, \neg A}$\\[16pt]
\end{tabular}
	
\noindent\begin{tabular}{ll}	
	$(\mathord{\wedge}\mathord{\Rightarrow\mid})$ \ $\dfrac{A, B,
		\Gamma \Rightarrow \Delta \mid S}{A\wedge B, \Gamma
		\Rightarrow \Delta \mid S}$ &
	 $(\mathord{\Rightarrow}\mathord{\wedge\mid})$ \ $\dfrac{\Gamma
		\Rightarrow \Delta, A \mid S \qquad
		\Gamma\Rightarrow \Delta,
		B \mid S}{\Gamma \Rightarrow \Delta,
		A\wedge B \mid S}$\\[16pt]
	
 $(\mathord{\mid\wedge}\mathord{\Rightarrow})$ \ $\dfrac{S \mid A, B,
		\Gamma \Rightarrow \Delta}{S \mid A\wedge B, \Gamma
		\Rightarrow \Delta}$ &
 $(\mathord{\mid\Rightarrow}\mathord{\wedge})$ \ $\dfrac{S\mid \Gamma
		\Rightarrow \Delta, A \qquad
		S\mid \Gamma\Rightarrow \Delta,
		B}{S\mid\Gamma \Rightarrow \Delta,
		A\wedge B }$\\[16pt]
\end{tabular}

\noindent\begin{tabular}{ll}	
	
 $(\mathord{\Rightarrow}\mathord{\vee}\mid)$ \ $\dfrac{\Gamma
		\Rightarrow \Delta, A,B \mid S} {\Gamma\Rightarrow \Delta,
		A\vee B \mid S}$ &
	 $(\mathord{\vee}\mathord{\Rightarrow}\mid)$ \ $\dfrac{A, \Gamma
		\Rightarrow \Delta \mid S \qquad B, \Gamma
		\Rightarrow \Delta \mid S}{A\vee B,
		\Gamma\Rightarrow \Delta \mid S}$\\[16pt]

	 $(\mathord{\mid\Rightarrow}\mathord{\vee})$ \ $\dfrac{S\mid\Gamma
		\Rightarrow \Delta, A,B } {S\mid\Gamma\Rightarrow \Delta,
		A\vee B}$ &
 $(\mathord{\mid\vee}\mathord{\Rightarrow})$ \ $\dfrac{S\mid A, \Gamma
		\Rightarrow \Delta \qquad S\mid B, \Gamma
		\Rightarrow \Delta }{S\mid A\vee B,
		\Gamma \Rightarrow \Delta}$
\end{tabular}
\end{center}
\begin{center}
 {\footnotesize    $(\mathord{\Rightarrow}\mathord{\rightarrow}\mid)$ \ $\dfrac{\Gamma
		\Rightarrow \Delta, B \mid A, \Pi\Rightarrow\Sigma \qquad \Gamma\Rightarrow\Delta, A \mid A,
		\Pi \Rightarrow \Sigma \qquad \Gamma\Rightarrow\Delta, B \mid B, \Pi\Rightarrow\Sigma} {\Gamma
		\Rightarrow \Delta, A\rightarrow B \mid \Pi\Rightarrow\Sigma}$}
\end{center}
\begin{center}
{\footnotesize     $(\mathord{\rightarrow}\mathord{\Rightarrow}\mid)$ \ $\dfrac{A, B, \Gamma
		\Rightarrow \Delta \mid \Pi\Rightarrow\Sigma \qquad \Gamma\Rightarrow\Delta \mid \Pi \Rightarrow \Sigma, A, B \qquad B, \Gamma\Rightarrow\Delta \mid \Pi\Rightarrow\Sigma, A} {A\rightarrow B, \Gamma
		\Rightarrow \Delta \mid \Pi\Rightarrow\Sigma}$}
\end{center}
\begin{center}
$(\mathord{\mid\rightarrow}\mathord{\Rightarrow})$ \, $\dfrac{\Gamma
	\Rightarrow \Delta, A \mid \Pi\Rightarrow\Sigma \quad \Gamma\Rightarrow\Delta \mid B,
	\Pi \Rightarrow \Sigma } {\Gamma
	\Rightarrow \Delta \mid A\rightarrow B, \Pi\Rightarrow\Sigma}$	\end{center}
	\begin{center}
    	$(\mathord{\mid\Rightarrow}\mathord{\rightarrow})$ \, $\dfrac{A, \Gamma
		\Rightarrow \Delta \mid\Pi\Rightarrow\Sigma, B}{\Gamma\Rightarrow\Delta \mid \Pi\Rightarrow\Sigma, A\rightarrow B}$
\end{center}	

All these rules satisfy the subformula property as well as the
properties of explicitness, separation, and symmetry. In particular, the rules are context-independent: no side conditions
are imposed, and consequently all rules are permutable.

A proof-search tree for a bisequent $B$ in \textbf{BSC} is a tree of bisequents with $B$ as the root and nodes generated by rules of \textbf{BSC}. A proof-search tree is axiomatic iff all leaves are axiomatic. We write  $\vdash_{\bf K_3^d} B$ if $B$ is $\Gamma \Rightarrow A \mid\; \Rightarrow$ and there is an axiomatic proof-search tree for it. We write $\vdash_{\bf LP^d} B$ if $B$ is  $\Rightarrow\; \mid \Gamma \Rightarrow A$ and there is an axiomatic proof-search tree for it. The height
of a proof-search tree is defined as the maximum length of its branches.

\begin{theorem}[Soundness and Completeness]
 For every finite multiset of formulas $\Gamma$ and every formula $A$, the following hold:
 \begin{itemize}
     \item $\Gamma \models_{\bf K_3^d} A$ iff $ \vdash_{\bf K_3^d} \Gamma \Rightarrow A \mid \;\Rightarrow$,
     \item $\Gamma \models_{\bf LP^d} A$ iff $ \vdash_{\bf LP^d} \;\Rightarrow\; \mid \Gamma \Rightarrow A $.
 \end{itemize}
\end{theorem}
\begin{proof}
    Follows from the results of \cite{IP1,IP2}.
\end{proof}

\begin{theorem}[Cut admissibility]
The following cut rules are admissible in \textbf{BSC}:
\begin{center}
	$(Cut\mid)$ \ $\dfrac{\Gamma 
	\Rightarrow  \Delta, A \mid \Lambda\Rightarrow\Theta \qquad
	A, \Pi  \Rightarrow  \Sigma \mid \Xi\Rightarrow\Omega} {\Gamma, \Pi  \Rightarrow  \Delta, \Sigma \mid \Lambda, \Xi\Rightarrow\Theta, \Omega}$
\end{center}

\begin{center}
	$(\mid Cut)$ \ $\dfrac{\Gamma 
	\Rightarrow  \Delta \mid \Lambda\Rightarrow\Theta, A \qquad
	\Pi  \Rightarrow  \Sigma \mid A, \Xi\Rightarrow\Omega} {\Gamma, \Pi  \Rightarrow  \Delta, \Sigma \mid \Lambda, \Xi\Rightarrow\Theta, \Omega}$
\end{center}
\end{theorem}
\begin{proof}
     Follows from \cite[Theorem 6.4]{IP2}.
\end{proof}

\begin{theorem}[Height-preserving admissibility of weakening and contraction] The following rules are height-preserving admissible in \textbf{BSC}:
    \begin{center}
        \noindent\begin{tabular}{ll}	
	$(\mathord{W}\mathord{\Rightarrow\mid})$ \ $\dfrac{\Gamma
		\Rightarrow \Delta \mid S}{A, \Gamma \Rightarrow \Delta \mid S}$ &
	$(\mathord{\Rightarrow}\mathord{W\mid})$ \ $\dfrac{\Gamma
		\Rightarrow \Delta \mid S}{\Gamma
		\Rightarrow \Delta, A \mid S}$\\[16pt]
	
	$(\mathord{\mid W}\mathord{\Rightarrow})$ \ $\dfrac{S \mid \Pi\Rightarrow\Sigma}{S\mid A, \Pi\Rightarrow\Sigma}$ &
	$(\mathord{\mid\Rightarrow}\mathord{W})$ \ $\dfrac{S \mid \Pi\Rightarrow\Sigma}{S \mid \Pi\Rightarrow\Sigma, A}$\\[16pt]	
\end{tabular}
\noindent\begin{tabular}{ll}	
	$(\mathord{C}\mathord{\Rightarrow\mid})$ \ $\dfrac{A,A,\Gamma
		\Rightarrow \Delta \mid S}{A, \Gamma \Rightarrow \Delta \mid S}$ &
	$(\mathord{\Rightarrow}\mathord{C\mid})$ \ $\dfrac{\Gamma
		\Rightarrow \Delta, A, A \mid S}{\Gamma
		\Rightarrow \Delta, A \mid S}$\\[16pt]
	
	$(\mathord{\mid C}\mathord{\Rightarrow})$ \ $\dfrac{S \mid A,A, \Pi\Rightarrow\Sigma}{S\mid A, \Pi\Rightarrow\Sigma}$ &
	$(\mathord{\mid\Rightarrow}\mathord{C})$ \ $\dfrac{S \mid \Pi\Rightarrow\Sigma, A, A}{S \mid \Pi\Rightarrow\Sigma, A}$\\[16pt]	
\end{tabular}
    \end{center}
\end{theorem}
\begin{proof}
    Follows from \cite[Theorem 6.1]{IP2} and \cite[Theorem 6.3]{IP2}.
\end{proof}

\begin{theorem}[Height-preserving invertibility]\label{Invertability}
For every instance of a logical rule, if its conclusion has a proof of height $n$, then each premise has a proof of height at most $n$.	
\end{theorem}
\begin{proof}
    Follows from \cite[Theorem 6.2]{IP2}.
\end{proof}

\section{Conclusion}\label{conclusion}
In this paper, we introduced two new three-valued logics, $ \mathbf{K_3^d}$ and $ \mathbf{LP^d}$, obtained by combining the conjunction and disjunction of $ \mathbf{K_3}$ and \textbf{LP} with weak Kleene--Bochvar implication and negation. These logics provide a simple many-valued setting in which variable-inclusion conditions can be seen as a variant of Epstein-style subject-matter restrictions. Although $ \mathbf{K_3^d}$ and $ \mathbf{LP^d}$ are not dependence logics in Epstein's original sense, they might be called dependence-like logics. We also provided sound, complete, and cut-free bisequent calculi for both logics using the uniform proof-theoretic framework developed in \cite{IP1,IP2}. One of the possible directions for further research is to extend $ \mathbf{K_3^d}$ and $ \mathbf{LP^d}$ by quantifiers and definite descriptions.

\paragraph{Acknowledgements.}  I would like to thank an anonymous referee for useful and constructive comments on an earlier version of the paper.  The research was funded by the European Union (ERC, ExtenDD, project number: 101054714). Views and opinions expressed are however those of the author(s) only and do not necessarily reflect those of the European Union or the European Research Council. Neither the European Union nor the granting authority can be held responsible for them.

\end{document}